\documentclass[11pt]{article}

\usepackage[preprint]{acl}

\usepackage{times}
\usepackage{latexsym}

\usepackage[T1]{fontenc}

\usepackage[utf8]{inputenc}

\usepackage{microtype}

\usepackage{inconsolata}

\usepackage{graphicx}

\usepackage{amsmath,amsfonts,amssymb,amsthm}
\usepackage{multirow, booktabs}
\usepackage{subcaption}
\usepackage{bm} 
\usepackage{hyperref}

\newtheorem{definition}{Definition}
\newtheorem{theorem}{Theorem}

\newtheorem{corollary}{Corollary}

\title{Statistical Foundations of DIME: Risk Estimation for Practical Index Selection}

\author{
  \textbf{Giulio D'Erasmo\textsuperscript{1}},
  \textbf{Cesare Campagnano\textsuperscript{3}},
  \textbf{Antonio Mallia\textsuperscript{3}},
\\
  \textbf{Pierpaolo Brutti\textsuperscript{1}},
  \textbf{Nicola Tonellotto\textsuperscript{2}},
  \textbf{Fabrizio Silvestri\textsuperscript{1}}
\\
\\
  \textsuperscript{1}Sapienza University of Rome,
  \textsuperscript{2}University of Pisa,
  \textsuperscript{3}Seltz, San Francisco, US
\\
  \small{
    \textbf{Correspondence:} \href{mailto:g.derasmo@diag.uniroma1.it}{g.derasmo@diag.uniroma1.it}
  }
}

\begin{document}
\maketitle

\begin{abstract}
High-dimensional dense embeddings have become central to modern Information Retrieval, but many dimensions are noisy or redundant. Recently proposed DIME (Dimension IMportance Estimation), provides query-dependent scores to identify informative components of embeddings. DIME relies on a costly grid search to select a priori a dimensionality for all the query corpus’s embeddings. Our work provides a statistically grounded criterion that directly identifies the optimal set of dimensions for each query at inference time. Experiments confirm achieving parity of effectiveness and reduces embedding size by an average of $\sim50\%$ across different models and datasets at inference time\footnote{Code available at \url{https://github.com/giulio-derasmo/RDIME}}.
\end{abstract}

\section{Introduction}
In Information Retrieval (IR) Dense Retrieval (DIR) represents a significant advancement that is driven by the use of high-dimensional embedding spaces to represent queries and documents. DIR models, for example~\citet{devlin-etal-2019-bert} and ~\citet{10.1145/3397271.3401075}, have demonstrated a competitive trade-off between effectiveness and efficiency in numerous benchmarks. However, observations have indicated that not all dimensions are equally useful, and some can even be harmful \cite{kovaleva-etal-2021-bert, puccetti-etal-2022-outlier, takeshita2025randomlyremoving50dimensions}. 

Research in this area has explored compressing encoder embeddings by projecting them onto a lower-dimensional manifold. For instance, one can use Principal Component Analysis (PCA) or autoencoders designed to learn and represent such a manifold \cite{liu-etal-2022-dimension, siciliano2024staticpruningdenseretrieval, 10.1145/3573128.3604896, 10.1145/3626772.3657765}. Although these approaches reduce retrieval performance, they do contribute to improved efficiency.

In contrast, \citet{10.1145/3626772.3657691} theorise the existence of a query-dependent lower-dimensional manifold within the original embedding space, where retrieval performance also increases. Their work introduces Dimension Importance Estimation (DIME), a method that scores each dimension by fusing the embeddings of the query and a relevant document through element-wise multiplication. These scores can then be used to rank dimensions, retaining only those most likely to contribute to effective retrieval. One major drawback of this approach is that the exact number of dimensions to retain cannot be fixed in advance, and previous work often resorts to evaluating a grid of candidate dimensionalities or simply keeping a fixed percentage of dimensions.

The current study addresses this limitation. In Section \ref{RQ1}, we prove that DIME estimates the squared latent signal, that is, the true information need underlying the query, which the embedding aims to capture. Unlike prior approaches, which do not specify the number of dimensions in advance, our Risk Dimension Importance Estimation (RDIME) provides a theoretically grounded criterion for identifying the dimensions to be retained. We further generalize DIME using a kernel-based formulation, which allows for rigorous weighted contributions from relevant documents and thereby improves the estimator's robustness. We validate this criterion in Section \ref{results}.

\section{Background}

\paragraph{DIME} In recent years, \citet{10.1145/3626772.3657691} conjectured the Manifold Clustering Hypothesis, which posits that high-dimensional embeddings lie in a query-dependent low-dimensional manifold where retrieval is more effective. 
To discover this manifold, Dimension Importance Estimation (DIME) assigns importance scores to each embedding dimension using a query-dependent function~$\bm{u}_q$. The importance scores should estimate the contribution of each dimension to retrieving relevant documents for the given query $\bm{q}$. 

They propose different methods for estimating the importance of dimensions in dense retrieval systems. These methods are based on reinforcing the query signal with a specific document $\bm{p}\! \in\! \mathbb{R}^p$ as: 
\begin{equation*}
    \bm{u}_q = \bm{q} \odot \bm{p}.
\end{equation*} 
Two examples are PRF and LLM DIMEs. The first one, inspired by \citet{rocchio1971relevance} Pseudo-Relevance Feedback (PRF), assumes that the top $M$ documents $\bm{d}_1, \ldots, \bm{d}_M$ retrieved by a similarity measure such as BM25 ~\cite{bm25} are likely relevant to the query. We shall refer to them as pseudo-relevant documents. Their centroid is used to reinforce the original signal. The second one, instead, uses a synthetic document generated by a Large Language Model (LLM).

Although PRF can enhance the retrieval effectiveness of IR systems, it suffers from several limitations. \citet{10.1145/3731120.3744579} address this by designing a scoring function that incorporates irrelevant feedback. In parallel work, \citet{10.1145/3726302.3730318} propose a different weighting scheme to avoid the non-tunable parameter~$k$.

\paragraph{Top-$k$ Thresholding} Following DIME, a Top-$k$ thresholding strategy is employed to select the most informative dimensions. Specifically, given the importance scores $\bm{u}_q$,
\begin{equation}\label{topk-thresholding}
    S = \operatorname{Top-}k ( \bm{u}_q) \subseteq \{1, \ldots,p\}
\end{equation}
where $S$ contains the indices of the $k$ dimensions with highest importance score in $\bm{u}_q$.

However, $k$ must be fixed globally for all queries and cannot be tuned using a validation set, limiting the method's adaptability.

\paragraph{Modulation Estimators} The recovery of true signals from noisy observations is a central problem in sparse and redundant representations theory \cite{beran1998modulation, elad2010sparse}. One approach to address this challenge is the modulation estimators framework. In this framework, a random vector $\bm{X}\! \in \!\mathbb{R}^p$ represents the noisy observation, where the dimensions $X_i$ are independent with $\mathbb{E}[X_i] = \xi_i$ and $\text{Var}(X_i)=\sigma^2$ for every $i \in \{1,\ldots,p\}$. Given a modulator function $(f_i)_{i=1}^p \in [0,1]^p$, the goal is to produce a linear estimator with components $f_iX_i$ for the corresponding true signal components $\xi_i$, such that the $\ell^2$ risk,  $\mathbb{E}\big[ \sum_{i=1}^p(\xi_i - f_iX_i)^2 \big]$, is minimised.

Under the assumption that only a few dimensions carry signal information, \citet{donoho1994ideal} proposed the hard-threshold estimator, defined component-wise as:
\begin{equation*}
    f_i X_i = \mathbf{1}_S(X_i) X_i ,  \quad S=\{x\in\mathbb{R} \ : \  |x|>\lambda \}
\end{equation*}
for some threshold $\lambda > 0$. 

We shall show that DIME can be interpreted as a modulation estimator of the query's latent signal, which represents a user's information need.

\section{Method}\label{RQ1}

In this section, we formalize the problem of estimating latent information needs from noisy query embeddings. We introduce a hard thresholding approach that retains only the most informative components, providing a practical strategy for denoising queries.

\subsection{IR for Modulation Estimators}

Let $\bm{\theta}\! \in\! \mathbb{R}^p$ denote the query's latent signal, or in other words, the user's information need. The query embedding $\bm{q}\! \in\! \mathbb{R}^p$,  which is obtained through an encoder, can be modeled as
\begin{equation*}
\bm{q} = \bm{\theta} + \varepsilon \bm{z}, \quad \bm{z}
\sim \mathcal{N}(0, I_p),
\end{equation*}
where $\varepsilon$ controls the noise level and $\bm{z}$ is standard Gaussian noise. The query embedding is then interpreted as a noisy sample of the underlying information need.

\begin{theorem}[Hard Thresholding Estimator]
Given the noisy embedding $\bm{q}$, the optimal hard-threshold estimator $\hat{\bm{\theta}}(S^\star)$ of the latent signal $\bm{\theta}$ under $\ell_2$ risk is given by: 
\begin{equation*}
\hat{\theta}_i(S^\star) = q_i \mathbf{1}_{S^\star}(i) \quad \forall  i \in \{1, \ldots, p\}, 
\end{equation*}
where 
\begin{equation}\label{Kset}
S^\star = \{i \in \{1, \ldots, p\}\,  | \, \theta^2_i > \varepsilon^2\}.
\end{equation}
\label{th:hte}
\end{theorem}
\begin{proof}
We want to find the best subset $S^\star$ that minimises the $\ell_2$ risk function:
\begin{align*}
S^\star &= \underset{S}{\mathrm{argmin}}\sum_{i=1}^p \mathbb{E} \big[ \hat{\theta}_i(S) - \theta_i \big]^2 \\
&= \underset{S}{\mathrm{argmin}} \Big\{ \sum_{i \in S} \mathbb{E} \big[ q_i - \theta_i \big]^2 + \sum_{i \notin S} \mathbb{E} \big[\theta_i^2 \big] \Big\} \\
&= \underset{S}{\mathrm{argmin}} \Big\{|S|\varepsilon^2 + \sum_{i \notin S} \theta_i^2 \Big\}
\end{align*}

\noindent
Each coordinate $i$ contributes $\varepsilon^2$ to the risk if retained, and $\theta_i^2$ if discarded. Therefore, to minimise the risk, we should include $i$ in $S^\star$ if and only if $\theta^2_i > \varepsilon^2$, which yields \eqref{Kset}.
\end{proof}

In order to be applied, Theorem~\ref{th:hte} requires knowledge of the set $S^\star$ defined in Eq.~\ref{Kset}, which depends on the unknown true signal $\bm{\theta}$. We show in the next section that DIME serves as an estimator of $\bm{\theta}^2$ and it can then be directly applied in Eq.~\ref{Kset}.

\subsection{DIME as an Estimator of Squared Latent Signal}

Let the query $\bm{q}\! \in\! \mathbb{R}^p$ and $\mathcal{D}^M = \{\bm{d}^{(i)}\}_{i=1}^M \subset \mathbb{R}^p$ its top-$M$ (pseudo) relevant documents. According to the Probability Ranking Principle \cite{robertson1977probability}, an IR system should rank documents in decreasing order of relevance to the query. Here, we interpret relevance as inversely related to variance: highly relevant documents are modeled as having lower noise. In comparison, less relevant documents have higher noise. Formally, we model each document $\bm{d}^{(i)}$ as:
\begin{equation*}
\bm{d}^{(i)} = \bm{\theta} + \sigma_i \bm{z}^{(i)}, \quad \bm{z}^{(i)}
\stackrel{i.i.d}{\sim} \mathcal{N}(0, I_p),
\end{equation*}
where $\sigma_1 \leq \sigma_2 \leq \ldots \leq \sigma_k$ reflects increasing uncertainty (or decreasing relevance) as the rank $i$ increases. By construction, higher-ranked documents are more likely to be truly relevant, and thus are distributed near the latent information need $\bm{\theta}$.


To formalize DIME as a statistical estimator with a controlled bias-variance trade-off, we introduce a kernel function $K(\bm{q}, \bm{d}^{(i)})$ that measures query documents similarity. This kernel determines how much each document contributes to the final estimate, with more similar documents receiving higher weight. This formulation allows us to generalize standard DIME variants within a unified framework called Kernel DIME.

\begin{definition}[Kernel DIME]
Kernel DIME is the solution of the local kernel-weighted least squares problem:
\begin{equation*}
\min_{\bm{u} \in \mathbb{R}^p} 
\sum_{i=1}^{M} w_i\, 
\big\| \bm{d}^{(i)} \odot \bm{q} - \bm{u} \big\|_2^2.
\end{equation*}
\noindent
where the weights $w_1, \ldots, w_k$ are given by
\begin{equation*}
w_i = \frac{K(\bm{q}, \bm{d}^{(i)})}{\sum_{j=1}^M K(\bm{q}, \bm{d}^{(j)})},
\end{equation*}
and $K: \mathbb{R}^p \times \mathbb{R}^p \to \mathbb{R}_+$ is a kernel function which measures the similarity between the query $\bm{q}$ and the documents $\bm{d}^{(i)}$.
\end{definition} 

In this case, Kernel DIME has a closed-form solution, which generalises DIMEs as:
\begin{equation}\label{kernel_dime}
\bm{u}_q = \bm{q} \odot \sum_{i=1}^{M} w_i\bm{d}^{(i)}. 
\end{equation}

\noindent Several meaningful choices of weights recover existing DIME variants: (i) \text{PRF DIME}, where $w_i = 1/k$; and (ii) \text{SWC DIME}, where for every dimension $i$, $w_i = \operatorname{softmax}(\bm{s})_i$; where $\bm{s}= \left[\bm{d}^{(i)}\bm{q}\right]_{i=1}^{M}$.

\begin{theorem}\label{prop-uniform-dime}
Kernel DIME with uniform weights $w_i = 1/k$ is an unbiased estimator of $\bm{\theta}^2$ for each component.
\end{theorem}
\begin{proof}
Using independence between $\bm{q}$ and $\bm{d}^{(i)}$s, we have:
\begin{align*}
\mathbb{E}[(\bm{u}_{q})_j] 
&= \sum_{i=1}^{M} \dfrac{1}{M} \mathbb{E}\big[d_j^{(i)} q_j\big] \\
&= \sum_{i=1}^{M} \dfrac{1}{M} \mathbb{E} [d_j^{(i)}]\mathbb{E}[q_j] = \theta_j^2.
\end{align*}
\end{proof} 


\noindent While uniform weighting yields unbiasedness, the choice of kernel influences the bias–variance tradeoff of Kernel DIME as an estimator of $\bm{\theta}^2$. Appendix~\ref{app:optimal_weights} discusses strategies for selecting kernel weights to control document noise aggregation.
It should be noted that we have not established formal theoretical guarantees for all kernel functions. Nevertheless, our experimental results in Section~\ref{results} suggest that several kernel choices perform well empirically.

\begin{corollary}[RDIME: Risk Dimension Importance Estimation]
Given a query $\bm{q}$, we can use the Kernel DIME representation $\bm{u}_q$, as defined in Eq.~\ref{kernel_dime}, for the Hard Thresholding Estimator described in Eq.~\ref{Kset}: 
\begin{equation*}
\begin{cases}
\hat{S} = \{i \in \{1, \ldots, p\}\,  | \, (\bm{u}_q)_i > \hat{\varepsilon}^2\} \\
\hat{\varepsilon}^2 = \frac{1}{p}\sum_{i=1}^p( q_i^2 - (\bm{u}_q)_i).
\end{cases}
\end{equation*}
\end{corollary}

\begin{table*}[ht]
\centering
\resizebox{\textwidth}{!}{%
\begin{tabular}{llccccccccccccccccc}
\toprule
\multirow{2}{*}{Model} & \multirow{2}{*}{Filter} &
\multicolumn{5}{c}{DL '19} &
\multicolumn{5}{c}{DL '20} &
\multicolumn{5}{c}{RB '04} \\
\cmidrule(lr){3-7} \cmidrule(lr){8-12} \cmidrule(lr){13-17}
 & & 0.4 & 0.6 & 0.8 & RDIME & $\Delta (\%)$ & 0.4 & 0.6 & 0.8 & RDIME & $\Delta (\%)$ & 0.4 & 0.6 & 0.8 & RDIME & $\Delta (\%)$ \\
\midrule
\multirow{3}{*}{ANCE} 
 & $\bm{u}^{LLM}$ & 0.570 & 0.660 & \textbf{0.663} & 0.656$^\star$ (0.94) & -1.05 & 0.533 & 0.622 & \textbf{0.658} & 0.651$^\star$ (0.93) & -1.06 & 0.257 & 0.381 & \textbf{0.392} & 0.387$^\star$ (0.94) & -1.27 \\
 & $\bm{u}^{PRF}$ & 0.552 & 0.640 & 0.657 & 0.650$^\star$ (0.93) & -1.06 & 0.541 & 0.616 & 0.650 & 0.645$^\star$ (0.93) & -0.92 & 0.269 & 0.357 & 0.383 & 0.386$^\star$ (0.92) & 0.78 \\
 & $\bm{u}^{SWC}$ & 0.558 & 0.645 & 0.653 & 0.652$^\star$ (0.93) & -0.15  & 0.546 & 0.616 & 0.651 & 0.645$^\star$ (0.93) & -0.77 & 0.265 & 0.361 & 0.381 & 0.384$^\star$ (0.92) & 0.79 \\
 & Baseline       & - & - & - & 0.645 (1.0) & 1.04 & - & - & - & 0.646 (1.0) & -0.19 & - & - & - & 0.384 (1.0) & -0.08 \\ 
 \midrule
\multirow{3}{*}{Contriever} 
 & $\bm{u}^{LLM}$ & 0.736 & 0.733 & \textbf{0.745} & 0.741$^\star$ (0.59) & -0.54 & 0.695 & 0.697 & 0.689 & 0.700$^\star$ (0.63) & 0.43  & \textbf{0.527} & 0.519 & 0.518 & 0.524$^\star$ (0.57) & -0.57 \\
 & $\bm{u}^{PRF}$ & 0.684 & 0.692 & 0.689 & 0.695$^\star$ (0.53) & 0.43  & 0.701 & \textbf{0.704} & 0.693 & 0.687 (0.55) & -2.41 & 0.489 & 0.491 & 0.492 & 0.481 (0.46) & -2.23 \\
 & $\bm{u}^{SWC}$ & 0.678 & 0.683 & 0.683 & 0.682$^\star$ (0.52) & -0.146 & 0.690 & 0.691 & 0.687 & 0.688$^\star$ (0.54) & -0.43 & 0.496 & 0.495 & 0.495 & 0.494$^\star$ (0.44) & -0.40 \\
 & Baseline       & - & - & - & 0.677 (1.0) & 0.87 & - & - & - & 0.666 (1.0) & 3.24 & - & - & - & 0.48 (1.0) & 2.79 \\
 \midrule
 \multirow{3}{*}{TAS-B} 
 & $\bm{u}^{LLM}$ & 0.763 & 0.758 & 0.757 & \textbf{0.766}$^\star$ (0.54) & 0.39  & 0.697 & 0.696 & 0.705 & 0.702$^\star$ (0.60) & -0.42 & 0.466 & 0.469 & 0.467 & \textbf{0.472}$^\star$ (0.54) & 0.64 \\
 & $\bm{u}^{PRF}$ & 0.737 & 0.738 & 0.735 & 0.738$^\star$ (0.51) & 0.00  & 0.705 & 0.712 & 0.706 & 0.707$^\star$ (0.52) & -0.70 & 0.442 & 0.448 & 0.448 & 0.444$^\star$ (0.43) & -0.89 \\
 & $\bm{u}^{SWC}$ & 0.729 & 0.732 & 0.728 & 0.726$^\star$ (0.51) & -0.68 & 0.712 & 0.717 & 0.712 & \textbf{0.720}$^\star$ (0.52) & 0.42  & 0.432 & 0.442 & 0.447 & 0.436 (0.44) & -2.46 \\
& Baseline       & - & - & - & 0.719 (1.0) & 1.02 & - & - & - & 0.685 (1.0) & 5.09 & - & - & - & 0.428 (1.0) & 1.87 \\
\bottomrule
\end{tabular}}
\caption{Comparison of nDCG@10 between RDIME and fixed Top-$k$ thresholding strategies ($k \in \{0.4, 0.6, 0.8\}$), across different query sets and DIR models. $\Delta (\%)$ represents the relative improvement of RDIME over the best performing Top-$k$ strategy. The value in parentheses indicates the average proportion of dimensions retained by RDIME. \textbf{Bold} indicates the best result per configuration. The superscript indicates that our criterion is not statistically significantly different from Top-$k$ thresholding methods.}
\label{results_table3}
\end{table*}

\noindent This corollary demonstrates that the hard thresholding estimator can be implemented directly from DIME scores. In contrast Eq. \ref{topk-thresholding}, which must explore a grid of candidate dimensionalities to determine how many dimensions to retain, RDIME provides a statistically grounded criterion for directly identifying the optimal dimension set. Crucially, this selection is performed in a query-dependent manner, allowing us to adapt the dimensionality per query rather than relying on a single fixed value for the entire collection (see Appendix~\ref{app:box_plots}).

\section{Experimental Section}
We follow prior work in DIME \cite{10.1145/3731120.3744579, 10.1145/3726302.3730318, 10.1145/3626772.3657691} to set up our experimental evaluation. 

\textbf{Datasets.} We evaluate on three passage retrieval benchmarks: TREC Deep Learning 2019 \citep[DL ’19]{trec19}, TREC Deep Learning 2020 \citep[DL ’20]{trec20}, and the Deep Learning Hard set \citep[DL HD]{trechd}. We assess out-of-domain robustness on TREC Robust 2004 collection \citep[RB ’04]{robust04}. 

\textbf{Models.} Experiments use three 768-dimensional DIR models: ANCE \cite{ance}, Contriever \cite{izacard2021unsupervised}, and TAS-B \cite{tasb}. We shall refer to these models at full dimensionality as Baseline.

\textbf{Evaluation.} We report mean Average Precision (AP) and nDCG@10. Statistical significance is tested using a paired Student’s $t$-test \cite{t-test} (or one-sided Wilcoxon signed-rank \cite{wilcoxon} if non-normal) at level $\alpha=0.05$ using Holm–Bonferroni corrections.

\textbf{DIMEs.} To validate our criterion, we use three DIME variants: PRF-DIME, LLM-DIME, and SWC-DIME. The hyperparameters are kept fixed, since the goal is not to tune performance but to demonstrate that our method correctly identifies the optimal number of dimensions. We set $M=2$ for PRF-DIME and $M=10$ for SWC-DIME.

%
%
%
%

\section{Results}\label{results}

Table~\ref{results_table3} reports nDCG@10 scores for DIME variants across collections and bi-encoders, comparing Top-$k$ thresholding ($k \in \{0.4, 0.6, 0.8\}$, as per existing literature), with our proposed RDIME method. The last column shows the improvement of our method over the best Top-$k$ baseline. The value in parentheses indicates the average proportion of dimensions retained by RDIME.

RDIME adapts dimensionality per query without hyperparameter tuning, whereas Top-$k$ methods (see \ref{topk-thresholding}) require a grid search to set the threshold. Since no validation set is available, this tuning is often done on the test set, and the optimal $k$ varies widely across collections and encoders. For example, $k=0.8$ in one setup and $k=0.4$ in another, leading to unpredictable performance.

Our criterion achieves performance that are not statistically different from the baselines in most settings. Where differences occur, they are small (0.15\% on DL ’19 to 2.46\% on RB ’04), demonstrating robust performance without hyperparameter tuning. Further results are provided in Appendix \ref{app:box_plots} (effects of query-specific dimension selection) and \ref{app:full_comparison} (AP and full dimensionality comparisons).


Taken together, these findings demonstrate that our theoretical criterion enables adaptive, query-specific dimensionality reduction without validation-set tuning, while maintaining performance competitive with prior method.

\section{Conclusion}


In this work, we introduced RDIME, a statistical criterion for query-dependent dimension selection in dense retrieval. Unlike previous approaches, which rely on grid search to set a single global dimensionality and cannot adapt at inference time, our method estimate the optimal dimensions for each query. We tested this framework using several DIMEs and observed consistent improvements in retrieval effectiveness, accompanied by substantial reductions in dimensionality.

This work opens the door to designing new DIMEs variants that leverage kernel functions to better structure the embedding space.

\section{Limitations}

We acknowledge two main limitations. First, our experimental evaluation focuses primarily on Standard DIMEs (PRF and LLM DIMEs), which can be seen as a special case of Kernel DIME. While Section~\ref{results} demonstrates the effectiveness of SWC DIME, the performance of other kernel functions (e.g., Radial Basis Function, Sigmoid) remains to be evaluated empirically.

Second, our theoretical analysis establishes unbiasedness only for uniform weights (Theorem~\ref{prop-uniform-dime}). Extending these guarantees to general kernel functions presents additional challenges. In particular, when weights depend on both the query and retrieved documents, they become random variables rather than fixed constants. Analyzing the statistical properties of Kernel DIME under such random weighting schemes, requires a more sophisticated theoretical treatment and is left for future work. Appendix~\ref{app:optimal_weights} provides preliminary guidance for kernel selection under the Probability Ranking Principle.

\bibliography{custom}

@inproceedings{10.1145/3731120.3744579,
author = {D'Erasmo, Giulio and Trappolini, Giovanni and Silvestri, Fabrizio and Tonellotto, Nicola},
title = {Eclipse: Contrastive Dimension Importance Estimation with Pseudo-Irrelevance Feedback for Dense Retrieval},
year = {2025},
isbn = {9798400718618},
publisher = {Association for Computing Machinery},
address = {New York, NY, USA},
url = {https://doi.org/10.1145/3731120.3744579},
doi = {10.1145/3731120.3744579},
booktitle = {Proceedings of the 2025 International ACM SIGIR Conference on Innovative Concepts and Theories in Information Retrieval (ICTIR)},
pages = {147–154},
numpages = {8},
location = {Padua, Italy},
series = {ICTIR '25}
}

@inproceedings{10.1145/3726302.3730318,
author = {Campagnano, Cesare and Mallia, Antonio and Silvestri, Fabrizio},
title = {Unveiling DIME: Reproducibility, Generalizability, and Formal Analysis of Dimension Importance Estimation for Dense Retrieval},
year = {2025},
isbn = {9798400715921},
publisher = {Association for Computing Machinery},
address = {New York, NY, USA},
url = {https://doi.org/10.1145/3726302.3730318},
doi = {10.1145/3726302.3730318},
booktitle = {Proceedings of the 48th International ACM SIGIR Conference on Research and Development in Information Retrieval},
pages = {3367–3376},
numpages = {10},
location = {Padua, Italy},
series = {SIGIR '25}
}

@inproceedings{10.1145/3626772.3657691,
author = {Faggioli, Guglielmo and Ferro, Nicola and Perego, Raffaele and Tonellotto, Nicola},
title = {Dimension Importance Estimation for Dense Information Retrieval},
year = {2024},
isbn = {9798400704314},
publisher = {Association for Computing Machinery},
address = {New York, NY, USA},
url = {https://doi.org/10.1145/3626772.3657691},
doi = {10.1145/3626772.3657691},
booktitle = {Proceedings of the 47th International ACM SIGIR Conference on Research and Development in Information Retrieval},
pages = {1318–1328},
numpages = {11},
location = {Washington DC, USA},
series = {SIGIR '24}
}

@article{donoho1994ideal,
  title={Ideal spatial adaptation by wavelet shrinkage},
  author={Donoho, David L and Johnstone, Iain M},
  journal={biometrika},
  volume={81},
  number={3},
  pages={425--455},
  year={1994},
  publisher={Oxford University Press}
}

@article{beran1998modulation,
  title={Modulation of estimators and confidence sets},
  author={Beran, Rudolf and D{\"u}mbgen, Lutz},
  journal={Annals of Statistics},
  pages={1826--1856},
  year={1998},
  publisher={JSTOR}
}

@book{elad2010sparse,
  title={Sparse and redundant representations: from theory to applications in signal and image processing},
  author={Elad, Michael},
  year={2010},
  publisher={Springer Science \& Business Media}
}

@article{bm25,
author = {Robertson, Stephen and Zaragoza, Hugo},
title = {The Probabilistic Relevance Framework: BM25 and Beyond},
year = {2009},
issue_date = {April 2009},
publisher = {Now Publishers Inc.},
address = {Hanover, MA, USA},
volume = {3},
number = {4},
issn = {1554-0669},
url = {https://doi.org/10.1561/1500000019},
doi = {10.1561/1500000019},
journal = {Found. Trends Inf. Retr.},
month = apr,
pages = {333–389},
numpages = {57}
}

@book{rocchio1971relevance,
  author    = {Rocchio, J.J.},
  title     = {Relevance Feedback in Information Retrieval},
  publisher = {Prentice Hall},
  address   = {Englewood Cliffs, New Jersey},
  year      = {1971}
}

@article{robertson1977probability,
  title={The probability ranking principle in IR},
  author={Robertson, Stephen E},
  journal={Journal of documentation},
  volume={33},
  number={4},
  pages={294--304},
  year={1977},
  publisher={MCB UP Ltd}
}

@misc{trec19,
      title={Overview of the TREC 2019 deep learning track}, 
      author={Nick Craswell and Bhaskar Mitra and Emine Yilmaz and Daniel Campos and Ellen M. Voorhees},
      year={2020},
      eprint={2003.07820},
      archivePrefix={arXiv},
      primaryClass={cs.IR},
      url={https://arxiv.org/abs/2003.07820}, 
}

@misc{trec20,
      title={Overview of the TREC 2020 deep learning track}, 
      author={Nick Craswell and Bhaskar Mitra and Emine Yilmaz and Daniel Campos},
      year={2021},
      eprint={2102.07662},
      archivePrefix={arXiv},
      primaryClass={cs.IR},
      url={https://arxiv.org/abs/2102.07662}, 
}

@inproceedings{trechd,
author = {Mackie, Iain and Dalton, Jeffrey and Yates, Andrew},
title = {How Deep is your Learning: the DL-HARD Annotated Deep Learning Dataset},
year = {2021},
isbn = {9781450380379},
publisher = {Association for Computing Machinery},
address = {New York, NY, USA},
url = {https://doi.org/10.1145/3404835.3463262},
doi = {10.1145/3404835.3463262},
booktitle = {Proceedings of the 44th International ACM SIGIR Conference on Research and Development in Information Retrieval},
pages = {2335–2341},
numpages = {7},
location = {Virtual Event, Canada},
series = {SIGIR '21}
}

@inproceedings{robust04,
  title={Overview of the TREC 2004 Robust Track},
  author={Voorhees, Ellen M.},
  booktitle={Proceedings of the Thirteenth Text REtrieval Conference (TREC 2004)},
  year={2004},
  organization={NIST Special Publication 500-261},
  address={Gaithersburg, MD},
  publisher={National Institute of Standards and Technology (NIST)}
}

@article{wilcoxon,
 ISSN = {00994987},
 URL = {http://www.jstor.org/stable/3001968},
 author = {Frank Wilcoxon},
 journal = {Biometrics Bulletin},
 number = {6},
 pages = {80--83},
 publisher = {[International Biometric Society, Wiley]},
 title = {Individual Comparisons by Ranking Methods},
 urldate = {2024-10-08},
 volume = {1},
 year = {1945}
}

@article{t-test,
 ISSN = {00063444, 14643510},
 URL = {http://www.jstor.org/stable/2331554},
 author = {Student},
 journal = {Biometrika},
 number = {1},
 pages = {1--25},
 publisher = {[Oxford University Press, Biometrika Trust]},
 title = {The Probable Error of a Mean},
 urldate = {2024-10-08},
 volume = {6},
 year = {1908}
}

@misc{tasb,
      title={Efficiently Teaching an Effective Dense Retriever with Balanced Topic Aware Sampling}, 
      author={Sebastian Hofstätter and Sheng-Chieh Lin and Jheng-Hong Yang and Jimmy Lin and Allan Hanbury},
      year={2021},
      eprint={2104.06967},
      archivePrefix={arXiv},
      primaryClass={cs.IR},
      url={https://arxiv.org/abs/2104.06967}, 
}

@inproceedings{
ance,
title={Approximate Nearest Neighbor Negative Contrastive Learning for Dense Text Retrieval},
author={Lee Xiong and Chenyan Xiong and Ye Li and Kwok-Fung Tang and Jialin Liu and Paul N. Bennett and Junaid Ahmed and Arnold Overwijk},
booktitle={International Conference on Learning Representations},
year={2021},
url={https://openreview.net/forum?id=zeFrfgyZln}
}

@article{izacard2021unsupervised,
  title={Unsupervised dense information retrieval with contrastive learning},
  author={Izacard, Gautier and Caron, Mathilde and Hosseini, Lucas and Riedel, Sebastian and Bojanowski, Piotr and Joulin, Armand and Grave, Edouard},
  journal={arXiv preprint arXiv:2112.09118},
  year={2021}
}

@inproceedings{devlin-etal-2019-bert,
    title = "{BERT}: Pre-training of Deep Bidirectional Transformers for Language Understanding",
    author = "Devlin, Jacob  and
      Chang, Ming-Wei  and
      Lee, Kenton  and
      Toutanova, Kristina",
    editor = "Burstein, Jill  and
      Doran, Christy  and
      Solorio, Thamar",
    booktitle = "Proceedings of the 2019 Conference of the North {A}merican Chapter of the Association for Computational Linguistics: Human Language Technologies, Volume 1 (Long and Short Papers)",
    month = jun,
    year = "2019",
    address = "Minneapolis, Minnesota",
    publisher = "Association for Computational Linguistics",
    url = "https://aclanthology.org/N19-1423/",
    doi = "10.18653/v1/N19-1423",
    pages = "4171--4186"
}

@inproceedings{10.1145/3397271.3401075,
author = {Khattab, Omar and Zaharia, Matei},
title = {ColBERT: Efficient and Effective Passage Search via Contextualized Late Interaction over BERT},
year = {2020},
isbn = {9781450380164},
publisher = {Association for Computing Machinery},
address = {New York, NY, USA},
url = {https://doi.org/10.1145/3397271.3401075},
doi = {10.1145/3397271.3401075},
booktitle = {Proceedings of the 43rd International ACM SIGIR Conference on Research and Development in Information Retrieval},
pages = {39–48},
numpages = {10},
location = {Virtual Event, China},
series = {SIGIR '20}
}

@inproceedings{kovaleva-etal-2021-bert,
    title = "{BERT} Busters: Outlier Dimensions that Disrupt Transformers",
    author = "Kovaleva, Olga  and
      Kulshreshtha, Saurabh  and
      Rogers, Anna  and
      Rumshisky, Anna",
    editor = "Zong, Chengqing  and
      Xia, Fei  and
      Li, Wenjie  and
      Navigli, Roberto",
    booktitle = "Findings of the Association for Computational Linguistics: ACL-IJCNLP 2021",
    month = aug,
    year = "2021",
    address = "Online",
    publisher = "Association for Computational Linguistics",
    url = "https://aclanthology.org/2021.findings-acl.300/",
    doi = "10.18653/v1/2021.findings-acl.300",
    pages = "3392--3405"
}

@inproceedings{puccetti-etal-2022-outlier,
    title = "Outlier Dimensions that Disrupt Transformers are Driven by Frequency",
    author = "Puccetti, Giovanni  and
      Rogers, Anna  and
      Drozd, Aleksandr  and
      Dell{'}Orletta, Felice",
    editor = "Goldberg, Yoav  and
      Kozareva, Zornitsa  and
      Zhang, Yue",
    booktitle = "Findings of the Association for Computational Linguistics: EMNLP 2022",
    month = dec,
    year = "2022",
    address = "Abu Dhabi, United Arab Emirates",
    publisher = "Association for Computational Linguistics",
    url = "https://aclanthology.org/2022.findings-emnlp.93/",
    doi = "10.18653/v1/2022.findings-emnlp.93",
    pages = "1286--1304"
}

@misc{takeshita2025randomlyremoving50dimensions,
      title={Randomly Removing 50\% of Dimensions in Text Embeddings has Minimal Impact on Retrieval and Classification Tasks}, 
      author={Sotaro Takeshita and Yurina Takeshita and Daniel Ruffinelli and Simone Paolo Ponzetto},
      year={2025},
      eprint={2508.17744},
      archivePrefix={arXiv},
      primaryClass={cs.LG},
      url={https://arxiv.org/abs/2508.17744}, 
}

@inproceedings{liu-etal-2022-dimension,
    title = "Dimension Reduction for Efficient Dense Retrieval via Conditional Autoencoder",
    author = "Liu, Zhenghao  and
      Zhang, Han  and
      Xiong, Chenyan  and
      Liu, Zhiyuan  and
      Gu, Yu  and
      Li, Xiaohua",
    editor = "Goldberg, Yoav  and
      Kozareva, Zornitsa  and
      Zhang, Yue",
    booktitle = "Proceedings of the 2022 Conference on Empirical Methods in Natural Language Processing",
    month = dec,
    year = "2022",
    address = "Abu Dhabi, United Arab Emirates",
    publisher = "Association for Computational Linguistics",
    url = "https://aclanthology.org/2022.emnlp-main.384/",
    doi = "10.18653/v1/2022.emnlp-main.384",
    pages = "5692--5698"
}

@misc{siciliano2024staticpruningdenseretrieval,
      title={Static Pruning in Dense Retrieval using Matrix Decomposition}, 
      author={Federico Siciliano and Francesca Pezzuti and Nicola Tonellotto and Fabrizio Silvestri},
      year={2024},
      eprint={2412.09983},
      archivePrefix={arXiv},
      primaryClass={cs.IR},
      url={https://arxiv.org/abs/2412.09983}, 
}

@inproceedings{10.1145/3573128.3604896,
author = {Acquavia, Antonio and Macdonald, Craig and Tonellotto, Nicola},
title = {Static Pruning for Multi-Representation Dense Retrieval},
year = {2023},
isbn = {9798400700279},
publisher = {Association for Computing Machinery},
address = {New York, NY, USA},
url = {https://doi.org/10.1145/3573128.3604896},
doi = {10.1145/3573128.3604896},
booktitle = {Proceedings of the ACM Symposium on Document Engineering 2023},
articleno = {7},
numpages = {10},
location = {Limerick, Ireland},
series = {DocEng '23}
}

@inproceedings{10.1145/3626772.3657765,
author = {Chang, Xuejun and Mishra, Debabrata and Macdonald, Craig and MacAvaney, Sean},
title = {Neural Passage Quality Estimation for Static Pruning},
year = {2024},
isbn = {9798400704314},
publisher = {Association for Computing Machinery},
address = {New York, NY, USA},
url = {https://doi.org/10.1145/3626772.3657765},
doi = {10.1145/3626772.3657765},
booktitle = {Proceedings of the 47th International ACM SIGIR Conference on Research and Development in Information Retrieval},
pages = {174–185},
numpages = {12},
location = {Washington DC, USA},
series = {SIGIR '24}
}

\appendix
\section{Appendix}

\begin{figure*}[ht]
    \centering

    \begin{subfigure}[b]{\textwidth}
        \centering
        \includegraphics[width=\textwidth]{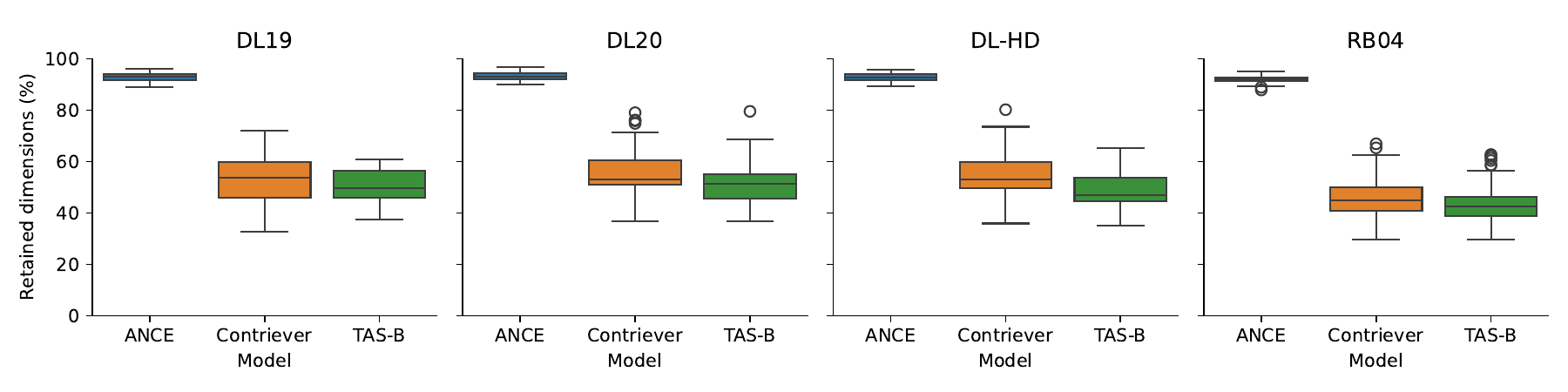}
        \caption{PRF@2 DIME}
        \label{fig:prf2}
    \end{subfigure}

    \vspace{0.5cm}  

    \begin{subfigure}[b]{\textwidth}
        \centering
        \includegraphics[width=\textwidth]{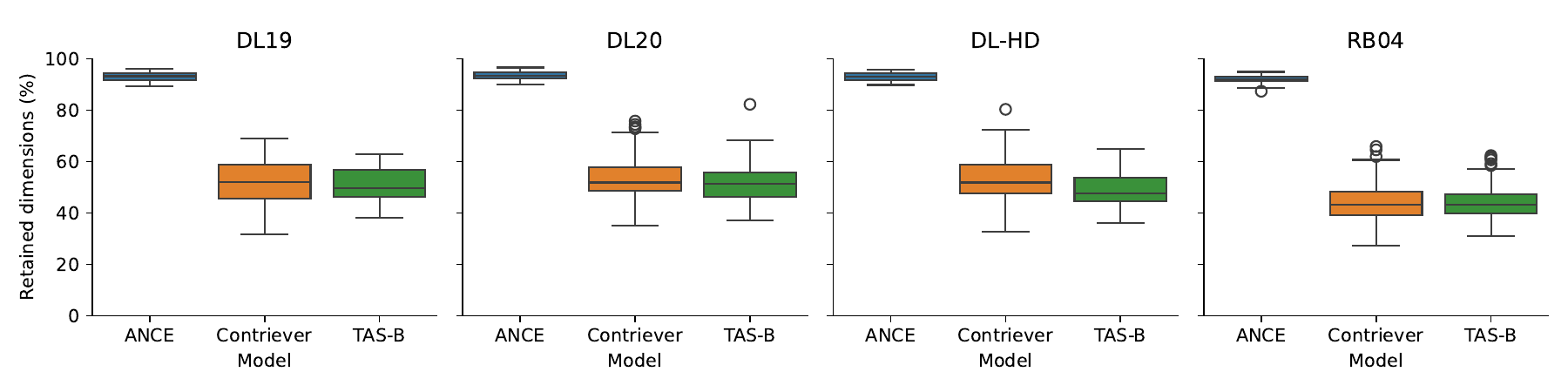}
        \caption{SWC@10 DIME}
        \label{fig:swc10}
    \end{subfigure}

    \vspace{0.5cm}  

    \begin{subfigure}[b]{\textwidth}
        \centering
        \includegraphics[width=\textwidth]{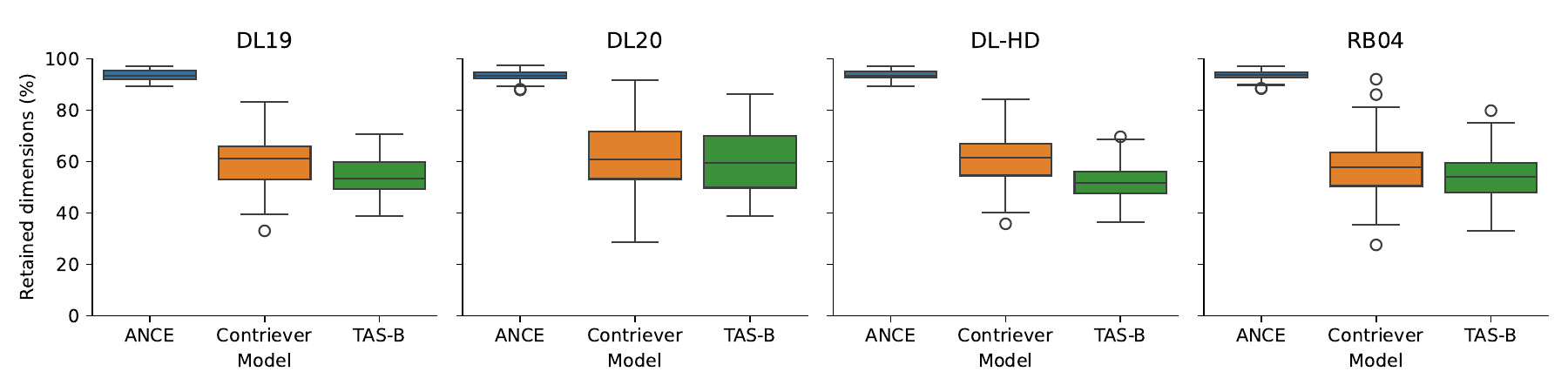}
        \caption{LLM DIME}
        \label{fig:gptfilter}
    \end{subfigure}

    \caption{Percentage of dimensions retained per query across three bi-encoders (ANCE, Contriever, TAS-B) and four collections (DL ’19, DL ’20, DL HD, RB ’04) using our risk-based criterion with three DIME variants. Boxplots show median, interquartile range, and outliers, revealing substantial query-dependent variation.}
    \label{fig:all_filters}
\end{figure*}

\subsection{Optimal weights under heteroskedastic noise in IR}\label{app:optimal_weights}
We aim to show that, as $M \to \infty$, the Kernel DIME estimator is not a consistent estimator, meaning that it does not converge in probability to the true unknown value $\bm{\theta}^2$ as the number of documents grows. Nevertheless, its Mean Squared Error (MSE) can be reduced through an appropriate choice of the weighting scheme, i.e, by selecting an optimal kernel function. For the sake of the proof, we assume that each weight $w_i$ is deterministic.

First, we can express the Kernel DIME explicitly as follows:
\begin{align*}
	(\bm{u}_q)_j &= q_j \sum_{i=1}^M w_i d_j^{(i)} \\
            &= (\theta_j + \varepsilon z_j)\Big[ \sum_{i=1}^M w_i (\theta_j + \sigma_i z^{(i)}_j) \Big] \\
            &= \theta_j^2 +  \theta_j \sum_{i=1}^M \sigma_iw_i z^{(i)}_j + \varepsilon \theta_j z_j \\
            &\quad +\varepsilon z_j \sum_{i=1}^M \sigma_iw_i z^{(i)}_j. 
\end{align*}
This allows us to plug the estimator into the MSE and, following the same reasoning as in Theorem \ref{prop-uniform-dime}, simplify the terms in the MSE, since the weights $w_i$ are constant. Hence,
\begin{align*}
\text{MSE}\Big((\bm{u}_q)_j\Big) &= \mathbb{V}\Big((\bm{u}_q)_j\Big) + \Big(\mathbb{E}[ (\bm{u}_q)_j ] - \theta_j^2 \Big)^2  \\
&= \mathbb{E}\Big[ \Big((\bm{u}_q)_j -  \theta_j^2\Big)^2 \Big] + 0\\
&= \mathbb{E}\Big[\Big(\theta_j \eta_{w,j} + \varepsilon \theta_j z_j + \varepsilon z_j \eta_{w,j} \Big)^2 \Big]
\end{align*}
where $\eta_{w,j} = \sum_{i=1}^M \sigma_iw_i \textbf{z}^{(i)}_j$.
Since $\mathbf{z}_j$ is independent of each $\textbf{z}^{(i)}_j$, we have 
\begin{align*}
\operatorname{cov}(z_j, \eta_{w,j}) &= \sum_{i=1}^M \sigma_iw_i \operatorname{cov}(z_j, z^{(i)}_j) =0.
\end{align*}
Furthermore, because $(\mathbf{z}_j, \eta_{w,j})$ is jointly gaussian, zero covariance implies independence. Consequently, $\textbf{z}_j$ is independent from $\eta_{w,j}$. By the standard property that measurable functions of independent random variables remain independent, it follows that $\textbf{z}_j$ is also independent of $\eta^2_{w,j}$. Therefore, 
\begin{equation*}
\mathbb{E}[z_j\eta^2_{w,j}] = \mathbb{E}[z_j]\mathbb{E}[\eta^2_{w,j}] = 0.
\end{equation*}
An identical argument shows that $z^2_j$ is independent of $\eta_{w,j}$. The same for $z^2_j$ and $\eta^2_{w,j}$.

Expanding the squared term in the MSE and removing the cross product using the argument we prove before,
\begin{align*}
\text{MSE}\Big((\bm{u}_q)_j \Big) &=
\theta_j^2\mathbb{E}[ \eta_{w,j}^2] + \varepsilon^2 \theta_j^2\mathbb{E}[ z_j^2] \\ &\quad \quad + \varepsilon^2 \mathbb{E}[z_j^2\eta^2_{w,j}]  \\
&=  \theta_j^2 \sum_{i=1}^M \sigma_i^2w_i^2 + \varepsilon^2 \theta_j^2 \\
&\quad \quad + \varepsilon^2 \mathbb{E}[z_j^2]\mathbb{E}[\eta^2_{w,j}] \\
&= (\theta_j^2 + \varepsilon^2)\sum_{i=1}^M \sigma_i^2 w_i^2 + \varepsilon^2 \theta_j^2.
\end{align*}

We can see that if we let the kernel weights to be uniform, the series $\lim_{M\to\infty} \sum_{i=1}^M \frac{\sigma^2_i}{M^2}$ does not converge. 

Instead we can minimize with respect to the weights resolving the optimization problem: 
\begin{equation*}
\min_{\bm{w} \in \mathbb{R}^p} f(\bm{w}) = \min_{\bm{w} \in \mathbb{R}^p}\sum_{i=1}^M \sigma_i^2w_i^2
\end{equation*}
restricted on $E = \{ \bm{w} \in \mathbb{R}^M \ |\ w_i \geq 0 \; \; \forall i  \, \, , \sum_{i=1}^M w_i = 1 \}$.
This is a well-posed problem, 
so we can derive the solution using the Method of Lagrange Multipliers having 
\begin{equation*}
w_i^{\star} = \dfrac{ \sigma_i^{-2} }{\sum_{j=1}^M \sigma_j^{-2}}.
\end{equation*}
With these weights, we conclude that for large $M$ the $\text{MSE}\Big((\bm{u}_q)_j\Big) \to \varepsilon^2 \theta_j^2$.

The optimal weights $w_i \propto \frac{1}{\sigma_i^2}$ demonstrate that effective DIME estimation requires inverse-variance weighting: documents with lower noise variance should contribute more heavily to the final estimate. However, practical implementations must operate without direct access to the true noise variances $\sigma_i^2$. 

This theoretical result provides a justification for the weighting scheme employed in SWC DIME. In this variant, documents with higher similarity scores to the query receive greater weight. Under the Probability Ranking Principle, higher-ranked (more similar) documents are assumed to have lower noise variance. Therefore, the SWC DIME weighting scheme can be interpreted as an approximation to inverse-variance weighting, where similarity scores serve as a proxy for the unobserved inverse noise variances.

\subsection{Importance of Query-Wise Dimension Selection}\label{app:box_plots}

To observe how important it is to select a unique set of dimensions for each query, Figure \ref{fig:all_filters} presents the distribution of retained dimensions across different bi-encoders and query sets when applying our criterion to each DIME variant. 

The boxplots reveal substantial variation in the optimal proportion of retained dimensions across queries, as reflected by the wide interquartile ranges and presence of outliers across all datasets and models. This demonstrates that query-dependent selection is necessary: while some queries require high dimensionality to preserve effectiveness, others can be adequately represented using approximately half the original dimensions.

The results reveal distinct behaviors across encoder models. ANCE consistently retains the highest proportion of dimensions (typically >90\%) regardless of dataset or DIME variant, suggesting that its representations distribute information more uniformly across the embedding space.  In contrast, Contriever and TAS-B are more prone to dimensional reduction, with the majority of queries retaining between 30\% and 70\% of dimensions. This indicates that these models encode information more sparsely, concentrating signal in a subset of dimensions.

\subsection{Full Comparison}\label{app:full_comparison}

In Table \ref{app:results_table2} we show
the complete results for all the collections, namely DL '19, DL '20, DL HD and RB '04, along with performance using AP e nDCG@10 metrics. We compare Top-$k$ thresholding ($k \in \{0.4, 0.6, 0.8\}$, as per existing literature), with our proposed RDIME method. The last column shows the improvement of our method over the best Top-$k$ baseline. The value in parentheses indicates the average proportion of dimensions retained by RDIME.

Our criterion achieves performances that are not statistically different from baselines in most settings. We notice that we can see a more consistent improvement in AP, proving the robustness of our method. In nDCG@10, where the differences occur, they are small.

In general, RDIME can be successfully used as a replacement for DIME-based methods, with the advantage of being hyperparameter-agnostic.

\begin{table*}[ht]
\centering
\resizebox{\textwidth}{!}{%
\begin{tabular}{llcccccccccccccccccccc}  
\toprule
\multirow{4}{*}{Model} & \multirow{4}{*}{Filter} &
\multicolumn{10}{c}{DL '19} &
\multicolumn{10}{c}{DL '20} \\
\cmidrule(lr){3-7} \cmidrule(lr){7-12} \cmidrule(lr){13-17} \cmidrule(lr){17-22}
& & & & AP &   & & & & nDCG@10 &   & & & & AP &  & & & & nDCG@10 \\
\cmidrule(lr){3-12} \cmidrule(lr){13-22}
 & & 0.4 & 0.6 & 0.8 & RDIME & $\Delta (\%)$ & 0.4 & 0.6 & 0.8 & RDIME & $\Delta (\%)$ & 0.4 & 0.6 & 0.8 & RDIME & $\Delta (\%)$ & 0.4 & 0.6 & 0.8 & RDIME & $\Delta (\%)$  \\
\midrule
\multirow{3}{*}{ANCE} 
 & $\bm{u}^{LLM}$ & 0.267 & 0.351 & \textbf{0.370} & 0.367 (0.94) & -1.0  & 0.570 & 0.660 & \textbf{0.663}$^\star$ & 0.656 (0.94)   & -1.19 & 0.281 & 0.372 & \textbf{0.397} & 0.395 (0.93) & -0.35 & 0.533 & 0.622 & \textbf{0.658} & 0.651$^\star$ (0.93) & -1.0 \\
 & $\bm{u}^{PRF}$ & 0.253 & 0.339 & \textbf{0.370} & 0.368 (0.93) & -0.54 & 0.552 & 0.640 & 0.657 & 0.650$^\star$ (0.93)   & -1.04 & 0.287 & 0.364 & 0.390 & 0.393 (0.93) & 0.92 & 0.541 & 0.616 & 0.651 & 0.645$^\star$ (0.93) & -0.94 \\
 & $\bm{u}^{SWC}$ & 0.255 & 0.340 & \textbf{0.370} & 0.369 (0.93)  & -0.14 & 0.558 & 0.645 & 0.653 & 0.652$^\star$ (0.93) & -0.09 & 0.289 & 0.366 & 0.390 & 0.392 (0.93) & 0.67 & 0.546 & 0.616 & 0.65 & 0.645$^\star$ (0.93) & -0.72 \\
 & Baseline         & - & - & - & 0.361 (1.0)             & 2.22  & - & - & - & 0.645 (1.0)            & 1.04  & - & - & - & 0.392 (1.0) & 0.15 & -   & - & - & 0.646 (1.0) & -0.19 \\
 \midrule
\multirow{3}{*}{Contriever} 
& $\bm{u}^{LLM}$ & 0.523 & 0.521 & 0.519 & \textbf{0.526} (0.59) & 0.65  & 0.736 & 0.733 & \textbf{0.745} & 0.741$^\star$ (0.59) & -0.51  & 0.500 & \textbf{0.505} & 0.501 & 0.503 (0.63) & -0.28   & 0.695 & 0.696 & 0.689 & 0.700$^\star$ (0.63) & 0.52 \\
& $\bm{u}^{PRF}$ & 0.507 & 0.511 & 0.509 & 0.513 (0.53) & 0.35  & 0.684 & 0.692 & 0.689 & 0.695$^\star$ (0.53) & 0.38   & 0.489 & 0.497 & 0.495 & 0.495 (0.55) & -0.42 & 0.701 & \textbf{0.704} & 0.693 & 0.687 (0.55) & -2.4 \\
& $\bm{u}^{SWC}$ & 0.512 & 0.513 & 0.509 & 0.514$^\star$ (0.52) & 0.21  & 0.678 & 0.683 & 0.683 & 0.682$^\star$ (0.52) & -0.06  & 0.497 & 0.495 & 0.495 & 0.497 (0.54) & 0.04  & 0.69 & 0.691 & 0.687 & 0.688 (0.54)$^\star$ & -0.48 \\
& Baseline         & - & - & - & 0.494 (1.0)            & 4.03  & - & - & - & 0.677 (1.0)              & 0.87   & - & - & - & 0.478 (1.0)                      & 3.85 & - & - & - & 0.666 (1.0) & 3.24 \\
\midrule
\multirow{3}{*}{TAS-B} 
& $\bm{u}^{LLM}$ & \textbf{0.527} & 0.52 & 0.514 & 0.526 (0.54)  & -0.25  & 0.763 & 0.758 & 0.757 & \textbf{0.766}$^\star$ (0.54)  & 0.42    & 0.495 & 0.494 & 0.495 & \textbf{0.497} (0.6) & 0.34 & 0.697 & 0.696 & 0.705 & 0.702$^\star$ (0.6) & -0.38 \\
& $\bm{u}^{PRF}$ & 0.509 & 0.512 & 0.507 & 0.511 (0.51) & -0.06  & 0.737 & 0.738 & 0.735 & 0.738$^\star$ (0.51)  & -0.01   & 0.486 & 0.489 & 0.489 & 0.490 (0.52) & 0.12 & 0.705 & 0.712 & 0.706 & 0.707$^\star$ (0.52) & -0.67 \\
& $\bm{u}^{SWC}$ & 0.508 & 0.506 & 0.503 & 0.508 (0.51) & -0.08   & 0.729 & 0.732 & 0.728 & 0.726$^\star$ (0.51) & -0.78   & 0.490 & 0.492 & 0.492 & 0.494 (0.52) & 0.24 & 0.712 & 0.717 & 0.712 & \textbf{0.720}$^\star$ (0.52) & 0.46 \\
& Baseline         & - & - & - & 0.476 (1.0)            & 6.68    & - & - & - & 0.719 (1.0)              & 1.02    & - & - & - & 0.476 (1.0) & 3.81             & - & - & - & 0.685 (1.0) & 5.09 \\
\midrule

\multirow{4}{*}{Model} & \multirow{4}{*}{Filter} &
\multicolumn{10}{c}{DL HD} &
\multicolumn{10}{c}{RB '04} \\
\cmidrule(lr){3-7} \cmidrule(lr){7-12} \cmidrule(lr){13-17} \cmidrule(lr){17-22}
& & & & AP &   & & & & nDCG@10 &   & & & & AP &  & & & & nDCG@10 \\
\cmidrule(lr){3-12} \cmidrule(lr){13-22}
 & & 0.4 & 0.6 & 0.8 & RDIME & $\Delta (\%)$ & 0.4 & 0.6 & 0.8 & RDIME & $\Delta (\%)$ & 0.4 & 0.6 & 0.8 & RDIME & $\Delta (\%)$ & 0.4 & 0.6 & 0.8 & RDIME & $\Delta (\%)$  \\
\midrule
\multirow{3}{*}{ANCE} 
& $\bm{u}^{LLM}$ & 0.125 & 0.172 & 0.186 & 0.183$^\star$ (0.94) & -1.35  & 0.253 & 0.329 & 0.346 & 0.336$^\star$ (0.94) & -2.66  & 0.082 & 0.137 & 0.148 & \textbf{0.149}$^\star$ (0.94) & 0.61  & 0.257 & 0.381 & \textbf{0.392} & 0.387$^\star$ (0.94) & -1.28 \\
& $\bm{u}^{PRF}$ & 0.128 & 0.176 & 0.18 & 0.183$^\star$ (0.93) & 1.39    & 0.278 & 0.340 & 0.335 & 0.331$^\star$ (0.93) & -2.53   & 0.084 & 0.134 & 0.148 & \textbf{0.149}$^\star$ (0.92) & 0.47 & 0.269 & 0.357 & 0.383 & 0.386$^\star$ (0.92) & 0.57 \\
& $\bm{u}^{SWC}$ & 0.125 & 0.174 & 0.184 & \textbf{0.185}$^\star$ (0.93) & 0.22   & 0.27 & \textbf{0.342} & 0.339 & 0.336$^\star$ (0.93) & -1.87   & 0.083 & 0.135 & 0.147 & \textbf{0.149}$^\star$ (0.92) & 1.57 & 0.265 & 0.361 & 0.381 & 0.384$^\star$ (0.92) & 0.95 \\
& Baseline       & - & - & - & 0.18 (1.0) & 2.38                 & - & - & - & 0.334 (1.0) & 0.51                & - & - & - & 0.146 (1.0) & 1.98       & - & - & - & 0.384 (1.0) & -0.08 \\
 \midrule
\multirow{3}{*}{Contriever} 
& $\bm{u}^{LLM}$ & 0.259 & \textbf{0.262} & 0.261 & 0.259$^\star$ (0.61) & -1.14  & 0.376 & 0.39 & 0.392 & 0.387$^\star$ (0.61) & -1.07   & \textbf{0.263} & 0.261 & 0.259 & \textbf{0.263}$^\star$ (0.57) & 0.08 & \textbf{0.527} & 0.519 & 0.518 & 0.524$^\star$ (0.57) & -0.4 \\
& $\bm{u}^{PRF}$ & 0.254 & 0.252 & 0.251 & 0.257$^\star$ (0.54) & 0.9    & \textbf{0.393} & 0.387 & 0.384 & \textbf{0.393}$^\star$ (0.54) & -0.05  & 0.254 & 0.256 & 0.257 & 0.253$^\star$ (0.46) & -1.4 & 0.489 & 0.491 & 0.492 & 0.481 (0.46) & -2.4 \\
& $\bm{u}^{SWC}$ & 0.253 & 0.248 & 0.249 & 0.25$^\star$ (0.53) & -0.99   & 0.387 & 0.383 & 0.382 & 0.383$^\star$ (0.53) & -0.98  & 0.256 & 0.257 & 0.258 & 0.257$^\star$ (0.44) & -0.43 & 0.496 & 0.495 & 0.495 & 0.494$^\star$ (0.44) & -0.4 \\
& Baseline       & - & - & - & 0.241 (1.0) & 3.73                & - & - & - & 0.375 (1.0) & 2.27                & - & - & - & 0.239 (1.0)              & 7.67 & - & - & - & 0.48 (1.0) & 2.79 \\
\midrule
\multirow{3}{*}{TAS-B} 
& $\bm{u}^{LLM}$ & 0.261 & \textbf{0.265} & 0.260 & 0.264$^\star$ (0.52)  & -0.19  & 0.402 & \textbf{0.408} & 0.395 & 0.405$^\star$ (0.52) & -0.71  & 0.214 & 0.218 & 0.218 & 0.218$^\star$ (0.54) & -0.09 & 0.466 & 0.469 & 0.467 & \textbf{0.472}$^\star$ (0.54) & 0.62 \\
& $\bm{u}^{PRF}$ & 0.243 & 0.250 & 0.256 & 0.238$^\star$ (0.49) & -6.77   & 0.383 & 0.388 & 0.394 & 0.382$^\star$ (0.49) & -3.05  & 0.219 & \textbf{0.224} & 0.222 & 0.220 (0.43)  & -1.74 & 0.442 & 0.448 & 0.448 & 0.444$^\star$ (0.43) & -0.87 \\
& $\bm{u}^{SWC}$ & 0.243 & 0.243 & 0.241 & 0.245$^\star$ (0.49) & 0.78   & 0.389 & 0.388 & 0.381 & 0.393$^\star$ (0.49) & 1.0    & 0.212 & 0.217 & 0.216 & 0.213 (0.44) & -1.8 & 0.432 & 0.442 & 0.447 & 0.436$^\star$ (0.44) & -2.46 \\
& Baseline       & - & - & - & 0.238 (1.0) & 2.85                & - & - & - & 0.374 (1.0) & 5.05                & - & - & - & 0.197 (1.0)              & 7.85 & - & - & - & 0.428 (1.0) & 1.87 \\
\bottomrule
\end{tabular}}
\caption{Comparison of retrieval effectiveness between RDIME and fixed Top-$k$ thresholding strategies ($k \in \{0.4, 0.6, 0.8\}$), across different query sets and DIR models. $\Delta (\%)$ represents the relative improvement of RDIME over the best performing Top-$k$ strategy. The value in parentheses indicates the average proportion of dimensions retained by RDIME. \textbf{Bold} indicates the best result per configuration. The superscript indicates that our criterion is not statistically significantly different from Top-$k$ thresholding methods.}
\label{app:results_table2}
\end{table*}

\end{document}